\documentclass[conference]{IEEEtran}
\IEEEoverridecommandlockouts
\usepackage{authblk}
\usepackage{cite,soul}
\usepackage{amsmath,amssymb,amsfonts,amsthm}
\usepackage{algorithm}
\usepackage{algorithmic}
\usepackage{graphicx}
\usepackage{textcomp}
\usepackage{xcolor}
\usepackage{wrapfig}
\newtheorem{theorem}{Theorem}
\newtheorem{definition}{Definition}
\newtheorem{proposition}[theorem]{Proposition}
\newtheorem{corollary}[theorem]{Corollary}
\newtheorem{lemma}{Lemma}

\usepackage{comment}
\usepackage{multirow}
\usepackage{booktabs}  
\usepackage{caption}
\usepackage{subcaption}
\usepackage{ctable}
\newtheorem{example}[theorem]{Example}

\newcommand{\Z}{\mathbb{Z}}

\def\BibTeX{{\rm B\kern-.05em{\sc i\kern-.025em b}\kern-.08em
    T\kern-.1667em\lower.7ex\hbox{E}\kern-.125emX}}

\newif\ifcomment
\commentfalse 

\begin{document}

\title{Quantum Matrix-Product Codes: CSS-T Characterization and Maximality} 

%

 \author[1]{Delio Jaramillo-Velez}
 \author[2]{Alessandro Neri}
 \author[3]{Adway Patra}
 \author[4]{Diego Ruano\thanks{\\ \textbf{E-mails}: djaramil@ull.edu.es, alessandro.neri@unina.it, apatra@umd.edu,\\ diego.ruano@uva.es, flavio.salizzoni@mis.mpg.de}}
 \author[5]{Flavio Salizzoni}

\affil[1]{Department of Mathematics, IMAULL, Universidad de La Laguna, Spain}
 \affil[2]{Department of Mathematics and Applications, University of Naples Federico II, Italy }
 \affil[3]{Department of Electrical and Computer Engineering, University of Maryland, United States }
 \affil[4]{IMUVa -- Mathematics Research Institute, Universidad de Valladolid, Spain}
 \affil[5]{Max Planck Institute for Mathematics in the Sciences, Germany }


\maketitle

\begin{abstract}
CSS-T codes are quantum error-correcting codes that play an important role in fault-tolerant quantum computation, as they help mitigate the proliferation of errors. They admit a transversal T-gate and are defined from a pair of nested classical binary linear codes satisfying a specific algebraic condition expressed in terms of their Schur square. We extend this algebraic characterization to the propagation rule known as the $(u \mid u+v)$-construction, that is, the matrix-product code constructed from two constituent codes. Moreover, for cyclic constituent codes, we provide an explicit characterization in terms of the defining cyclotomic sets, which extends existing results for cyclic codes. This framework allows the construction of new and longer CSS-T codes.


\end{abstract}

\section{Introduction}

Quantum error-correcting codes that admit transversal implementations of non-Clifford gates are central to fault-tolerant quantum computation. 
Among non-Clifford gates, the $T$ gate occupies a distinguished role: together with the Clifford group, it generates a universal gate set, and its fault-tolerant realization is one of the main structural constraints in quantum code design. 
CSS-T codes were introduced as Calderbank--Shor--Steane (CSS) stabilizer codes \cite{CSS1,CSS2} admitting a transversal physical $T$ gate acting as a logical operator, or as the identity up to phase \cite{rengaswamy2020classical}. 
Such codes are optimal among nondegenerate stabilizer codes supporting transversal $T$ in a precise sense \cite{rengaswamy2020classical,CSST2}, and they provide an algebraic framework for constructing magic-state-free or magic-state-efficient fault-tolerant schemes.

Algebraically, a CSS-T code is determined by a pair of nested binary linear codes $(C_1,C_2)$ satisfying additional compatibility conditions involving componentwise, or Schur, products. 
A structural breakthrough was the algebraic characterization of binary CSS-T pairs in terms of Schur squares \cite{algebraic_CSS-T}. 
More precisely, a pair $(C_1,C_2)$ gives rise to a CSS-T code if and only if
\[
C_2 \subseteq C_1 \cap (C_1^{\star 2})^\perp,
\]
where $C_1^{\star 2}$ denotes the Schur square of $C_1$. 
This characterization reveals that transversal-$T$ compatibility is controlled by quadratic constraints on classical codes. 
It also led to explicit constructions and classifications in important families, notably cyclic codes, for which the CSS-T condition admits a description in terms of generator polynomials and defining sets \cite{algebraic_CSS-T}. 
These results connect the theory of CSS-T codes to the algebraic study of Schur products of linear codes.

From a classical coding theory perspective, matrix-product codes form a powerful and flexible method for constructing long linear codes from shorter constituent codes. 
Given linear codes $C_1,\dots,C_s$ over $\mathbb{F}_q$ and an $s\times m$ matrix $M$ over $\mathbb{F}_q$, the matrix-product code
\[
C = [C_1,\dots,C_s]M
\]
consists of all matrix products of row vectors from the $C_i$ with $M$~\cite{blackmore2001matrix}. 
This construction generalizes classical schemes such as the $(u \mid u+v)$ (Plotkin) construction and admits strong control over dimension and minimum distance; see \cite{blackmore2001matrix,ozbudak2002note,hernando2009construction}. 
Matrix-product codes have been used extensively to construct classical codes with good parameters and desirable algebraic structure, including nested and cyclic families.

A key result for our purposes is the explicit description of Schur squares of matrix-product codes \cite{square_matrix_prod}. 
Under suitable conditions on the defining matrix $M$, the square $C^{\star 2}$ of a matrix-product code can itself be described as a matrix-product code whose constituents are sums of pairwise Schur products $C_i \star C_j$. 
For instance, in the $(u \mid u+v)$-construction, corresponding to the matrix
\begin{equation}\label{eq:matrixA}
\mathcal{A} :=\begin{pmatrix}1 & 1\\ 0 & 1\end{pmatrix},
\end{equation}the Schur square 
of the $(u \mid u+v)$ matrix-product code is again a $(u \mid u+v)$ matrix-product code whose constituents are described explicitly in terms of $C_1^{\star 2}$ and the mixed product $C_1\star C_2$; see Proposition \ref{chara_squar_pro}. 
This explicit control of $C^{\star 2}$ is particularly well suited for analyzing the CSS-T condition, since the CSS-T condition reduces to orthogonality constraints with $C_1^{\star 2}$ and $C_1\star C_2$, making the quadratic transversal-$T$ condition directly tractable at the level of the original codes $C_1$ and $C_2$.

Our contributions can be summarized as follows:
\begin{enumerate}
    \item Let $C_2\subseteq C_1$ be matrix-product codes obtained from the $(u \mid u+v)$-construction. We provide necessary and sufficient conditions for  $(C_1,C_2)$ to be a CSS-T pair. 
   We also determine when such pairs are maximal in the CSS-T poset.
    \item We specialize the framework to the case in which the constituent codes of $C_1$ and $C_2$ are cyclic.
    In this setting, we express the CSS-T and maximality conditions in terms of the defining cyclotomic sets associated with the codes.
    This yields explicit combinatorial criteria on cyclotomic sets governing when a pair of cyclic matrix-product codes satisfies the CSS-T condition and when it gives rise to a maximal CSS-T pair.
\end{enumerate}
 



\smallskip



\smallskip

\noindent
\textbf{Organization.}
The paper is organized as follows. 
In Section~\ref{sec:preliminaries}, we review CSS-T pairs, Schur products, and matrix-product codes. 
Section~\ref{sec:condition_and_max} treats the CSS-T conditions for matrix-product codes under the $(u \mid u+v)$-construction, and the conditions for maximality. 
Section~\ref{sec:cyclic} specializes the theory to cyclic constituent codes.

\section{Preliminaries}\label{sec:preliminaries}

Throughout this paper, all classical linear codes are binary linear codes, that is, linear subspaces of $\mathbb{F}_2^n$. We omit explicit reference to the base field unless necessary. 

\subsection{Linear codes and duality}

Let $C \subseteq \mathbb{F}_2^n$ be a linear code. The \emph{dimension} of $C$ as a $\mathbb{F}_2$-vector space is denoted by $\dim(C)$. The \emph{minimum distance} of $C$ is
    \[
        \delta(C) = \min\{\mathrm{wt}(c) : c \in C \setminus \{0\}\},
    \] where $\mathrm{wt}(c)$ denotes the Hamming weight of $c$. The \emph{dual code} of $C$ is defined by 
    \[
        C^\perp = \{ x \in \mathbb{F}_2^n : x \cdot c = 0 \text{ for all } c \in C \}.
    \]
    
A binary code $C$ is called \emph{even} if every codeword has even Hamming weight. 
Equivalently, $C$ is even if and only if $C \subseteq \langle \mathbf{1} \rangle^\perp$, 
where $\mathbf{1} = (1,\dots,1)$.

\subsection{Schur (componentwise) product of codes}

For vectors $x = (x_1,\dots,x_n)$ and $y = (y_1,\dots,y_n)$ in $\mathbb{F}_2^n$, 
their \emph{Schur product} (or componentwise product) is defined as
\[
    x \star y = (x_1 y_1,\dots,x_n y_n).
\]
Given linear codes $C, D \subseteq \mathbb{F}_2^n$, their Schur product is
\[
    C \star D = \langle c \star d : c \in C,\ d \in D \rangle.
\]
The \emph{square} of a code $C$ is denoted by $C^{\star 2} := C \star C$.

Since we work over $\mathbb{F}_2$, for every $c \in C$, we have 
$c \star c = c$, and therefore $C \subseteq C^{\star 2}$.
\subsection{CSS and CSS-T pairs}

Let $C_2 \subseteq C_1 \subseteq \mathbb{F}_2^n$ be linear codes. 
The classical CSS construction produces a quantum code with parameters
\[
    [[n, \dim(C_1) - \dim(C_2), d]],
\]
with $
    d = \min\{ \mathrm{wt}(C_1 \setminus C_2),\ \mathrm{wt}(C_2^\perp \setminus C_1^\perp) \}$, see~\cite{CSS1,CSS2}.

CSS-T quantum codes are defined by a pair of classical binary linear codes $(C_1,C_2)$, called a CSS-T pair. They were characterized in~\cite{algebraic_CSS-T}, as stated in the following result.

\begin{theorem}[{\cite[Theorem 2.3]{algebraic_CSS-T}\label{charac_CSS-T}}]
A pair $(C_1, C_2)$ of binary linear codes is a \emph{CSS-T pair} if and only if
$$
C_2 \subseteq C_1,\hbox{ and } C_1^{\star 2} \subseteq C_2^\perp.
$$
\end{theorem}
Moreover, for a  CSS-T pair $(C_1, C_2)$, we have that $\min \{\delta(C_1),\delta(C_2^\perp)\}=\delta(C_2^\perp)$, and the parameters of the corresponding CSS-T code are  (\cite[Corollary 2.5]{algebraic_CSS-T}) $$[[n, \dim(C_1) - \dim(C_2),\ge {\delta}(C_2^\perp)]].$$
In particular, if $(C_1, C_2)$ is a CSS-T pair, then $C_2$ is self-orthogonal.
This Schur-based characterization will be the main tool for establishing CSS-T 
conditions for matrix-product codes.
\subsection{Matrix-product codes}
We begin with the general definition of matrix-product codes.
\begin{definition}[Matrix-Product Code, \cite{blackmore2001matrix}]
Let $C_1,\dots,C_s \subseteq \mathbb{F}_2^n$ be linear codes and let 
$M \in \mathbb{F}_2^{s \times \ell}$ be a matrix with rank $s$ (implying $s \leq \ell$). 
The \emph{matrix-product code} defined by the constituent codes $C_1,\dots,C_s$ and $M$ is
\[
C = [C_1,\dots,C_s]M
= \{ [\mathbf{c}_1,\dots,\mathbf{c}_s]M :
\mathbf{c}_i \in C_i \}.
\]
\end{definition}
Equivalently, identifying each codeword with a vector in $\mathbb{F}_2^{n\ell}$ by reading the entries of the $s \times \ell$ matrix in column-major order, 
a codeword of $C$ can be written in block form as
\[
\left(
\sum_{i=1}^s m_{i1}\mathbf{c}_i,\,
\sum_{i=1}^s m_{i2}\mathbf{c}_i,\,
\dots,\,
\sum_{i=1}^s m_{i\ell}\mathbf{c}_i
\right) \in \mathbb{F}_2^{n\ell},
\]
where $M = (m_{ij})$.

Under the natural hypothesis that the matrix $M$ has full rank, the dimension of $C$ is given by $\dim(C)=\dim(C_1)+\cdots+\dim(C_s)$. Moreover, if the codes are nested, $C_1 \supseteq \cdots \supseteq C_s$, then the minimum distance is given by
\[
\min \{d_1 \delta(C_1), \ldots, d_s \delta(C_s)\},
\]
where $d_i$ is the minimum distance of the code generated by the first $i$ rows of $M$.

Throughout this paper, we work over $\mathbb{F}_2$ and focus mainly on the 
$(u \mid u+v)$-construction, which corresponds to the case $s=\ell=2$ and the matrix $M=\mathcal{A}$  given in \eqref{eq:matrixA}.


In this case, for linear codes $C_1, C_2 \subseteq \mathbb{F}_2^n$, the matrix-product code is
\[
[C_1, C_2]\mathcal{A}
=
\{ (c_1,\, c_1 + c_2) :
c_1 \in C_1,\; c_2 \in C_2 \}
\subseteq \mathbb{F}_2^{2n},
\] and has parameters $$[2n,\dim(C_1)+\dim(C_2),\min\{2\delta(C_1),\delta(C_2)\}].$$

The next results allow us to translate CSS-T conditions for matrix-product codes into
Schur-product conditions on their constituent codes.
\begin{proposition}[{\cite[Proposition 6.2]{blackmore2001matrix}}]\label{dual_matrix_pro}
    Let $C=[C_1,C_2]\mathcal{A}$ be a matrix-product code. Then
    $
    C^{\perp}=[C_1^{\perp},C_{2}^{\perp}](\mathcal{A}^{-1})^{T}.
    $
\end{proposition}

\begin{proposition}[{\cite[Theorem 3.1]{square_matrix_prod}}]
\label{prop:product_mp_code}
Let $D_1,D_2,D_1',D_2'$ be linear codes and let
$
C_1=[D_1,D_2]\mathcal{A}$, $C_2=[D_1',D_2']\mathcal{A}$.
Then,
\[
C_1\star C_2
=
[D_1\star D_1',\,
D_1\star D_2'
+
D_2\star D_1'
+
D_2\star D_2']\mathcal{A}.
\]
\end{proposition}

\begin{proposition}[{\cite[Corollary 3.2]{square_matrix_prod}}]\label{chara_squar_pro}
    Let $C_1$ and $C_2$ be linear codes, and let $C=[C_1,C_2]\mathcal{A}$ be its matrix-product code. Then, $C^{\star2}=[C_1^{\star2}, (C_1+C_2)\star C_2]\mathcal{A}$.
\end{proposition}
We conclude this section with an example illustrating a CSS-T code built from a matrix-product code.

\begin{example}[A CSS-T example arising from a matrix-product code]
Let $n=2$, $D_1:=\langle (1,1)\rangle$ and $D_2:=\langle (0,0)\rangle$. 
Define $C_1 := [D_1,D_2]\mathcal{A}$ and $C_2 := C_1.$
Then
\[
C_1 = \{(d,\ d): d\in D_1\} = \langle (1,1,1,1)\rangle.
\]
Since 
$
(C_1^{\star 2})^\perp = C_1^\perp
= \{x\in \mathbb{F}_2^4:\mathrm{wt}(x)\equiv 0 \pmod 2\},
$
we have $C_1^{\star 2}=C_1$ and
$C_1\subseteq (C_1^{\star 2})^\perp$. Hence $
C_2 \subseteq C_1 \cap (C_1^{\star 2})^\perp,
$
so $(C_1,C_2)$ is a CSS-T pair built from a matrix-product code.
\end{example}

\section{CSS-T conditions for matrix-product codes}\label{sec:condition_and_max}

In this section, we determine necessary and sufficient conditions for a pair of matrix-product codes to form a CSS-T pair.


Our first main result is formulated as follows.

\begin{theorem}\label{CSS-T_matrix_product}
Let $D_1,D_2,D_1',D_2'$ be linear codes such that $
D_2 \subseteq D_1
\quad\text{and}\quad
D_2' \subseteq D_1'$.
Define
\[
C_1 = [D_1,D_2]\mathcal{A},
\qquad
C_2 = [D_1',D_2']\mathcal{A}.
\]
Then, assuming $C_2 \subseteq C_1$, the pair $(C_1,C_2)$ is CSS-T if and only if $
D_1^{\star 2} \subseteq (D_2')^\perp
$, and $D_1 \star D_2 \subseteq (D_1')^\perp.$
\end{theorem}
Note that the assumption $C_2 \subseteq C_1$ holds if $D_1' \subseteq D_1$ and $D_2' \subseteq D_2$. To prove Theorem~\ref{CSS-T_matrix_product} we first establish the following lemma.
\begin{lemma}\label{lemma:mp_subset_equiv}
Let $E_1,E_2,E_1',E_2'$ be linear codes such that
$
E_2 \subseteq E_1$, and $E_1' \subseteq E_2'.
$ Then the following conditions are equivalent:
$$
1)\: [E_1,E_2]\mathcal{A}
    \subseteq[E_1',E_2'](\mathcal{A}^{-1})^T,\:2)\:E_1 \subseteq E_2'\hbox{ and } E_2 \subseteq E_1'.
$$
\end{lemma}

\begin{proof}
Assume first $1)$. If
$e_1 \in E_1$ and $e_2 \in E_2$, then there exist 
$e_1' \in E_1'$ and $e_2' \in E_2'$ such that
$
(e_1,\,e_1+e_2)
=
(e_1'+e_2',\,e_2'),
$ since 
\begin{equation}\label{eq:matrixA-1}
 \mathcal{A}^{-1} =
 \begin{pmatrix}
 1 & 1 \\
 0 & 1
 \end{pmatrix}.
\end{equation}
Hence $
e_1 = e_1' + e_2'$, and
$e_1 + e_2 = e_2'$.
Since $e_1' \in E_1' \subseteq E_2'$ and $e_2' \in E_2'$, we get $
e_1 = e_1' + e_2' \in E_2'$.
Moreover, $
e_2 = (e_1+e_2) + e_1
= e_2' + (e_1' + e_2')
= e_1' \in E_1'$.
Thus $E_1 \subseteq E_2'$ and $E_2 \subseteq E_1'$. Conversely, assume that $E_1 \subseteq E_2'$ and $E_2 \subseteq E_1'$. 
Given $(e_1,e_1+e_2)$ with $e_1\in E_1$ and $e_2\in E_2$, define $e_1' := e_2$, and $e_2' := e_1 + e_2$. Then, $e_1'\in E_1'$ and $e_2'\in E_2'$, and
\[
(e_1'+e_2',\,e_2')
=
(e_2 + (e_1+e_2),\,e_1+e_2)
=
(e_1,\,e_1+e_2).
\]
Hence the desired inclusion holds.
\end{proof}
We are now ready to prove Theorem~\ref{CSS-T_matrix_product}.
\begin{proof}[Proof of Theorem \ref{CSS-T_matrix_product}]
By Theorem~\ref{charac_CSS-T}, the pair $(C_1,C_2)$ is CSS-T if and only if $
C_2 \subseteq C_1
\quad\text{and}\quad
C_1^{\star 2} \subseteq C_2^\perp$.
By assumption, $C_2 \subseteq C_1$, so it remains to characterize the condition $
C_1^{\star 2} \subseteq C_2^\perp.
$
From Proposition~\ref{chara_squar_pro} we have
\[
C_1^{\star 2}
=
[D_1^{\star 2},\ (D_1+D_2)\star D_2]\mathcal{A},
\]
and from Proposition~\ref{dual_matrix_pro},
$
C_2^\perp
=
[D_1'^{\perp}, D_2'^{\perp}](\mathcal{A}^{-1})^T.
\qquad
$
Applying Lemma~\ref{lemma:mp_subset_equiv} with
\[
E_1 = D_1^{\star 2}, 
\;\;
E_2 = (D_1+D_2)\star D_2, \;\;  E_1'=D_1'^\perp, \;\; E_2'=D_2'^\perp,
\]
we obtain that $
C_1^{\star 2}\subseteq C_2^\perp
$
is equivalent to $
D_1^{\star 2} \subseteq D_2'^{\perp}
\quad\text{and}\quad
(D_1+D_2)\star D_2 \subseteq D_1'^{\perp}.
$
Finally, since $D_2 \subseteq D_1$, we have $D_1 + D_2 = D_1$, and hence
$
(D_1+D_2)\star D_2 = D_1 \star D_2.
$
This proves the equivalence.
\end{proof}

\subsection{Maximality in the CSS-T poset}

We now study maximal elements in the poset of CSS-T pairs, focusing on matrix-product codes.

\begin{definition}[{\cite{algebraic_CSS-T}}] 
Let $\mathcal{P}$ denote the poset of CSS-T pairs ordered by
$
(C_1,C_2) \le (C_1',C_2')$ if and only if 
$C_i \subseteq C_i'$ for $i=1,2.$
Let $\mathcal{M} \subseteq \mathcal{P}$ be the subposet consisting of CSS-T pairs
$(C_1,C_2)$ such that both $C_1$ and $C_2$ are matrix-product codes.
\end{definition}

\begin{definition}
Let $(C_1,C_2)\in\mathcal{M}$.
\begin{itemize}
    \item We say that $(C_1,C_2)$ is \emph{maximal in $C_1$} if 
    $(C_1,C_2)\le (C_1',C_2)$ with $C_1'$ a matrix-product code 
    implies $C_1'=C_1$.
    \item Similarly, $(C_1,C_2)$ is \emph{maximal in $C_2$} if 
    $(C_1,C_2)\le (C_1,C_2')$ with $C_2'$ a matrix-product code 
    implies $C_2'=C_2$.
\end{itemize}
\end{definition}

We next give a characterization of equality for the matrix-product codes considered in CSS-T pairs.

\begin{lemma}\label{lemma:mp_equivalence}
Let $E_1,E_2,E_1',E_2'$ be linear codes such that
$
E_2 \subseteq E_1$, and $
E_1' \subseteq E_2'.
$ Then the following are equivalent:
$$
1)\:[E_1,E_2]\mathcal{A}=[E_1',E_2'](\mathcal{A}^{-1})^T,\: 2)\: E_1 = E_2'\hbox{ and }E_2 = E_1'.
$$

\end{lemma}

\begin{proof}
The equality holds if and only if both inclusions hold.
First, assume $[E_1,E_2]\mathcal{A}
\subseteq
[E_1',E_2'](\mathcal{A}^{-1})^T$.
By Lemma~\ref{lemma:mp_subset_equiv}, this is equivalent to $
E_1 \subseteq E_2',$ 
$E_2 \subseteq E_1'$.
Similarly, the reverse inclusion $
[E_1',E_2'](\mathcal{A}^{-1})^T
\subseteq
[E_1,E_2]\mathcal{A}$
is equivalent (again by Lemma~\ref{lemma:mp_subset_equiv}) to
$E_2' \subseteq E_1$ and $
E_1' \subseteq E_2$. Combining the two inclusions yields $
E_1 = E_2'$ and $E_2 = E_1'$.
Conversely, if $E_1 = E_2'$ and $E_2 = E_1'$, both inclusions hold,
and therefore the matrix-product codes are equal.
\end{proof}

As a consequence, in Corollary~\ref{cor:mp_maximal_conditions} we extend the following result about maximal pairs in the poset.

\begin{theorem}[{\cite[Theorem 3.11]{algebraic_CSS-T}}]
\label{thm:csst_maximal}
Let $C_2 \subset C_1 \subset \mathbb{F}_2^n$ be linear codes.  
The pair $(C_1,C_2)$ is maximal in $\mathcal{P}$ if and only if
$
C_1^{\perp} = C_1 \star C_2\hbox{ and }C_2^{\perp} = C_1^{\star 2}.
$
\end{theorem}

\begin{corollary}\label{cor:mp_maximal_conditions}
Let $(C_1,C_2)$ be a CSS-T pair with
\[
C_1=[D_1,D_2]\mathcal{A},
\qquad
C_2=[D_1',D_2']\mathcal{A},
\]
where $D_2\subseteq D_1$ and $D_2'\subseteq D_1'$.  
Then $(C_1,C_2)$ is maximal in $\mathcal{M}$ if and only if
\begin{align*}
C_1^\perp
&=
[D_1^\perp,D_2^\perp](\mathcal{A}^{-1})^T
=
C_1\star C_2,\\
C_2^\perp
&=
[D_1'^\perp,D_2'^\perp](\mathcal{A}^{-1})^T
=
C_1^{\star 2}.
\end{align*}
\end{corollary}

\begin{proof}
The result follows by combining Theorem~\ref{thm:csst_maximal}
with Proposition~\ref{dual_matrix_pro},
Proposition~\ref{prop:product_mp_code},
and Proposition~\ref{chara_squar_pro}.
\end{proof}

We can also state the previous result in terms of the constituent codes.

\begin{theorem}\label{thm:mp_maximal_conditions_components}
Let $(C_1,C_2)$ be a CSS-T pair with
\[
C_1=[D_1,D_2]\mathcal{A},
\qquad
C_2=[D_1',D_2']\mathcal{A},
\]
where $D_2\subseteq D_1$ and $D_2'\subseteq D_1'$.  
Then $(C_1,C_2)$ is maximal in $\mathcal{M}$ if and only if
\begin{enumerate}
    \item $D_1^{\perp}
    =
    D_1\star D_2'
    +
    D_2\star D_1'
    +
    D_2\star D_2'$,
    \item $D_2^{\perp}
    =
    D_1\star D_1'$,
    \item $D_1'^{\perp}
    =
    D_1\star D_2$,
    \item $D_2'^{\perp}
    =
    D_1^{\star 2}$.
\end{enumerate}
\end{theorem}

\begin{proof}
Apply Lemma~\ref{lemma:mp_equivalence} to the equalities
in Corollary~\ref{cor:mp_maximal_conditions}.
\end{proof}

The following result gives propagation rules in the poset of CSS-T pairs.

\begin{proposition}\label{prop:mp_extensions}
Let $(C_1,C_2)$ be a CSS-T pair with
\[
C_1=[D_1,D_2]\mathcal{A},
\qquad
C_2=[D_1',D_2']\mathcal{A},
\]
where $D_2\subseteq D_1$ and $D_2'\subseteq D_1'$.
Let $y\in\mathbb{F}_2^n$. Then the following hold:

\begin{enumerate}
\item[(a)]
$([D_1+\langle y\rangle,D_2]\mathcal{A},[D_1',D_2']\mathcal{A})$ is CSS-T if and only if $
D_1\star y + \langle y\rangle \subseteq (D_2')^\perp$, and $y\star D_2 \subseteq (D_1')^\perp$.
\item[(b)] If $y\in D_1$,
$([D_1,D_2+\langle y\rangle]\mathcal{A},[D_1',D_2']\mathcal{A})$ is CSS-T if and only if $
y\star D_1 \subseteq (D_1')^\perp$.
\item[(c)]
If $y\in D_1$, $([D_1,D_2]\mathcal{A},[D_1'+\langle y\rangle,D_2']\mathcal{A})$ is CSS-T if and only if $
y \in (D_1\star D_2)^\perp$.
\item[(d)]
If $y\in D_1'\cap D_2$, $([D_1,D_2]\mathcal{A},[D_1',D_2'+\langle y\rangle]\mathcal{A})$ is CSS-T if and only if $y \in (D_1^{\star2})^\perp$.
\end{enumerate}
\end{proposition}

\begin{proof}
By Theorem~\ref{CSS-T_matrix_product}, a pair
$
([E_1,E_2]\mathcal{A},[F_1,F_2]\mathcal{A})
$
is CSS-T if and only if $
E_1^{\star2} \subseteq F_2^\perp
\quad\text{and}\quad
E_1\star E_2 \subseteq F_1^\perp.
$

Since $(C_1,C_2)$ is CSS-T, we already have
\[
D_1^{\star2} \subseteq (D_2')^\perp,
\qquad
D_1\star D_2 \subseteq (D_1')^\perp.
\]
\smallskip
(a) Replacing $D_1$ by $D_1+\langle y\rangle$, the CSS-T conditions become
\[
(D_1+\langle y\rangle)^{\star2} \subseteq (D_2')^\perp,
\quad
(D_1+\langle y\rangle)\star D_2 \subseteq (D_1')^\perp.
\]
Now $
(D_1+\langle y\rangle)^{\star2}
=
D_1^{\star2}
+
D_1\star y
+
\langle y\rangle,
$
and
\[
(D_1+\langle y\rangle)\star D_2
=
D_1\star D_2
+
y\star D_2.
\]
Since the original pair is CSS-T, the conditions reduce to
\[
D_1\star y +
\langle y\rangle \subseteq (D_2')^\perp
\quad\text{and}\quad
y\star D_2 \subseteq (D_1')^\perp.
\]
\smallskip
(b) Replacing $D_2$ by $D_2+\langle y\rangle$, we obtain
\[
D_1^{\star2} \subseteq (D_2')^\perp,
\quad
D_1\star (D_2+\langle y\rangle)
\subseteq (D_1')^\perp.
\]
Expanding the second condition gives
\[
D_1\star D_2 + D_1\star y
\subseteq (D_1')^\perp.
\]
Since the original pair is CSS-T, this reduces to $
y\star D_1 \subseteq (D_1')^\perp$.
\smallskip

(c) Replacing $D_1'$ by $D_1'+\langle y\rangle$, the second CSS-T condition becomes $
D_1\star D_2
\subseteq (D_1'+\langle y\rangle)^\perp.$
Since $
(D_1'+\langle y\rangle)^\perp
=
D_1'^\perp \cap \langle y\rangle^\perp, $
this is equivalent to
\[
D_1\star D_2 \subseteq D_1'^\perp
\quad\text{and}\quad
y \in (D_1\star D_2)^\perp.
\]
The first condition already holds, so we obtain the desired result.
\smallskip

(d) Replacing $D_2'$ by $D_2'+\langle y\rangle$, the first CSS-T condition becomes
$
D_1^{\star2}
\subseteq (D_2'+\langle y\rangle)^\perp
=
D_2'^\perp \cap \langle y\rangle^\perp.$
Since $D_1^{\star2}\subseteq D_2'^\perp$ already holds, this is equivalent to $
y \in (D_1^{\star2})^\perp.
$
\end{proof}





\section{Cyclic codes as constituent codes}\label{sec:cyclic}

We now focus on matrix-product codes whose constituent codes are cyclic.
In~\cite{square_matrix_prod}, distance properties of matrix-product codes
with cyclic constituent codes were investigated, while in~\cite{algebraic_CSS-T}
CSS-T codes were constructed from cyclic codes using algebraic methods.
In this section, we combine these two perspectives to study CSS-T codes
arising from matrix-product codes with cyclic constituents.
A linear code $C\subseteq \mathbb{F}_2^n$ is called \emph{cyclic} if it is invariant under cyclic shifts, i.e., whenever $(c_0,c_1,\ldots,c_{n-1})\in C$, then
$(c_{n-1},c_0,\ldots,c_{n-2})\in C.$
Equivalently, $C$ is an ideal of
$
R=\mathbb{F}_2[x]/\langle x^n-1\rangle
$
generated by a polynomial $g(x)$ such that $g(x)\mid x^n-1$.
Assume that $n$ and $q=2$ are coprime, $\mathbb{Z}_n:=\{1,\dots,n\}$, and let $\beta$ be a primitive $n$-th root of unity in an algebraic closure of $\mathbb{F}_2$. We use the notation $C(I)$ to denote a cyclic code with  \emph{generating set}
$
I=\{j\in \mathbb Z_n: g(\beta^j)\neq 0\}.$ 
{It is well-known that $I$ must be a \emph{cyclotomic set}, that is, a subset of $\Z_n$ such that $
I=q\cdot I,
$
where $
q\cdot I:=\{q\cdot i : i\in I\}$, with $q=2$. Under these assumptions, if $S,T\subset \mathbb Z_n$ are two cyclotomic sets, it holds that 
\begin{equation}\label{eq:containment_cosets}
    C(S)\subseteq C(T) \;\Longleftrightarrow \; S\subseteq T.
\end{equation}
}

\begin{proposition}[{\cite{square_matrix_prod}}]\label{dual_cyclic}
    Let $I\subseteq \mathbb{Z}_{n}$ be a cyclotomic set. Then
    $C(I)^{\perp}= C(-J)$, where $J=\mathbb{Z}_{n}\setminus I$.
\end{proposition}

\begin{proposition}[{\cite[Proposition 4.3]{square_matrix_prod}}]\label{sum_prod_cyclic}
 Let $I_1,I_2\subseteq \mathbb{Z}_{n}$ be cyclotomic sets. Then
 $$
 C(I_1)+C(I_2)=C(I_1\cup I_2),\quad C(I_1)\star C(I_2)=C(I_1+I_2).
 $$
\end{proposition}

We now present our first result describing when a pair of $(u \mid u+v)$ matrix-product codes
with \emph{cyclic} constituent codes forms a CSS-T pair. We use the above description
of cyclic codes in terms of cyclotomic sets.

\begin{theorem}\label{Th:cyclic}
Let $I_1,I_2,I_1',I_2'$ be cyclotomic sets in $\mathbb{Z}_n$ such that $I_2\subseteq I_1$,  $I_2'\subseteq I_1'$, $I_1'\subseteq I_1$, $I_2'\subseteq I_2.$
Then the pair $
\bigl([C(I_1),C(I_2)]\mathcal{A},\ [C(I_1'),C(I_2')]\mathcal{A}\bigr)
$
is a CSS-T pair if and only if
$$
    I_1+I_1\subseteq -J_2',\hbox{ and }I_1+I_2\subseteq -J_1',
$$
or equivalently, if and only if
$$
   n\notin I_1+I_1+I_2',\hbox{ and } n\notin I_1+I_2+I_1'.
$$
\end{theorem}


\begin{proof}
Apply Theorem~\ref{CSS-T_matrix_product} with
\[
D_1=C(I_1),\quad D_2=C(I_2),\quad D_1'=C(I_1'),\quad D_2'=C(I_2').
\]
Since $I_1'\subseteq I_1$ and $I_2'\subseteq I_2$, by \eqref{eq:containment_cosets} we have
$C(I_1')\subseteq C(I_1)$ and $C(I_2')\subseteq C(I_2)$, hence $C_2\subseteq C_1$.
Therefore the CSS-T condition is equivalent to
\[
C(I_1)^{\star 2}\subseteq C(I_2')^\perp,
\qquad
C(I_1)\star C(I_2)\subseteq C(I_1')^\perp.
\]
By Propositions~\ref{sum_prod_cyclic} and~\ref{dual_cyclic},
\[
C(I_1)^{\star 2}=C(I_1+I_1),
\qquad
C(I_1)\star C(I_2)=C(I_1+I_2),
\]
and $
C(I_2')^\perp = C(-J_2'),$ $
C(I_1')^\perp = C(-J_1')$.
Thus the CSS-T condition becomes
\[
C(I_1+I_1)\subseteq C(-J_2'),
\qquad
C(I_1+I_2)\subseteq C(-J_1').
\]
 Hence, by \eqref{eq:containment_cosets} the above is equivalent to 
$I_1+I_1\subseteq -J_2'$,  $I_1+I_2\subseteq -J_1'.$
Finally, using the standard relation
\[
S\subseteq -J \quad\Longleftrightarrow\quad 0\notin S+I
\quad\Longleftrightarrow\quad
n\notin S+I,
\]
where $J=\mathbb{Z}_n\setminus I$, we obtain the equivalent conditions
\[
n\notin I_1+I_1+I_2',
\qquad
n\notin I_1+I_2+I_1'.\qedhere
\]\end{proof}

\begin{example}
 Consider $n=15$, and $I_1=\{1,2,4,8,15\}$, $I_2=I_1'=I_2'=\{1,2,4,8\}$. Given that  
 $$
 15\notin I_1+I_1+I_2',\quad 15\notin I_1+I_2+I_1',
 $$
 from Theorem~\ref{Th:cyclic} and {\rm \cite[Corollary 4.4]{square_matrix_prod}} we obtain a CSS-T code with parameters $[[30,1,\geq3]]$ given by $C_1,C_2$, where
 $$
 C_1=[C(I_1),C(I_2)]\mathcal{A},\quad C_2=[C(I_1'),C(I_2')]\mathcal{A}.
 $$
\end{example}

We now characterize maximal CSS-T pairs in the setting of cyclic matrix-product codes.

\begin{theorem}\label{thm:cyclic_maximal}
Let $I_1,I_2,I_1',I_2'$ be cyclotomic sets in $\mathbb{Z}_n$ such that $
I_2\subseteq I_1$, $I_2'\subseteq I_1'$, $I_1'\subseteq I_1$,  $I_2'\subseteq I_2$, 
and suppose that $
(C_1,C_2)
=
\bigl([C(I_1),C(I_2)]\mathcal{A},\ [C(I_1'),C(I_2')]\mathcal{A}\bigr)
$
is a CSS-T pair. For $i=1,2$, denote by $
J_i := \mathbb{Z}_n \setminus I_i$,  $J_i' := \mathbb{Z}_n \setminus I_i'$
the complements of the defining sets.
Then $(C_1,C_2)$ is maximal in $\mathcal{P}$ if and only if
\begin{enumerate}
    \item $(I_1+I_2')\cup (I_2+I_1')\cup (I_2+I_2') = -J_1,$
    \item $I_1+I_1' = -J_2,$
     \item $I_1+I_2 = -J_1',$
    \item $I_1+I_1 = -J_2'.$
\end{enumerate}
\end{theorem}

\begin{proof}
By Theorem~\ref{thm:mp_maximal_conditions_components}, the pair $(C_1,C_2)$
is maximal if and only if
\begin{align*}
    D_1^{\perp}=
D_1\star D_2'
+
D_2\star& D_1'
+
D_2\star D_2',\quad
D_2^{\perp}
=
D_1\star D_1',\\
D_1'^{\perp}
=
D_1&\star D_2,\quad
D_2'^{\perp}
=
D_1^{\star2}
\end{align*}


where $D_1=C(I_1)$, $D_2=C(I_2)$,
$D_1'=C(I_1')$, and $D_2'=C(I_2')$.
Using Proposition~\ref{sum_prod_cyclic}, we have
\[
C(S)\star C(T)=C(S+T),
\qquad
C(S)^{\star2}=C(S+S),
\]
and by Proposition~\ref{dual_cyclic}, $
C(I)^\perp = C(-J),$ where  $J=\mathbb{Z}_n\setminus I.$
Thus the four equalities above translate into
\begin{align*}
C(-J_1)
&=
C(I_1+I_2')
+
C(I_2+I_1')
+
C(I_2+I_2'),\\
C(-J_2)
&=
C(I_1+I_1'),\:
C(-J_1')
=
C(I_1+I_2),\\
C(-J_2')
&=
C(I_1+I_1).
\end{align*}
Since sums of cyclic codes correspond to unions of defining sets,
and equality of cyclic codes corresponds to equality of defining sets,
the result follows.
\end{proof}

\section{Conclusions and Future Work}
In this work, we analyzed CSS-T conditions for matrix-product codes arising from the
$(u \mid u+v)$-construction. We obtained necessary and sufficient conditions for a pair
of such matrix-product codes $(C_1,C_2)$ to be a CSS-T pair, and we characterized
maximality of these pairs in the CSS-T poset. We further translated these conditions
to cyclic constituent codes, obtaining explicit cyclotomic-set criteria for CSS-T and
maximal CSS-T pairs. These results provide a concrete bridge between the algebraic
structure of matrix-product codes and the CSS-T condition, giving practical criteria
that can be used to construct and classify new families of quantum codes supporting
transversal T-gates. Future work will include extending this analysis to weighted
Reed--Muller codes as constituent codes in the matrix-product construction, as well as
a systematic study of the parameters obtained from both the cyclic and weighted
Reed--Muller settings~\cite{reed-muller}.

\section*{Acknowledgments}
This project began during the “Coding Theory and Cryptography Summer School and Collaboration Workshop,” which took place in July 2024 at the Stager Center for International Scholarship – Virginia Tech in Switzerland. The authors sincerely thank the organizers for creating an environment that made this collaboration possible. This work was partially supported by NSF grant DMS-240155, Grant PID2022-138906NB-C21 funded by MCIN/AEI/ 10.13039/501100011033 and by ERDF/UE, and by Grant CLU-2025-1-02- IMUVA funded by Department of Education of the Junta de Castilla y Le\'on and FEDER Funds, the INdAM - GNSAGA Project CUP E53C24001950001, and by the P500PT-222344 SNSF project.


\IEEEtriggeratref{13}
\bibliographystyle{ieeetr}
\bibliography{references}


\end{document}